\documentclass[11pt,a4paper]{article}
\usepackage[ansinew]{inputenc}
\usepackage[T1]{fontenc}
\usepackage{amsthm}
\usepackage{amssymb}
\usepackage{amsmath}
\usepackage{amsfonts}

\usepackage{epsfig}
\usepackage{graphicx}

\usepackage{graphicx}
\usepackage{pstcol,pstricks,pst-node,pst-text,pst-3d}

\newtheorem{theorem}[equation]{Theorem}

\newtheorem{lemma}[equation]{Lemma}

\newtheorem{example}[equation]{Example}

\theoremstyle{definition}
\newtheorem{definition}[equation]{Definition}
\newtheorem{remark}[equation]{Remark}

\numberwithin{equation}{section}

\newcommand{\R}{{\mathbb{R}}}

\begin{document}

\title{On the Localization of the Personalized PageRank of Complex Networks}

\date{}

\author{
E.\,Garc\'{\i}a$^1$, F.\,Pedroche$^2$ and M.\,Romance$^{1,3}$\\
\\
{\small $^1$\,Departamento de Matem\'{a}tica Aplicada} \\
{\small ESCET - Universidad Rey Juan Carlos} \\
{\small C/ Tulip\'an s/n 28933 M\'ostoles (Madrid), Spain.}\\
 \\
{\small $^2$\,Institut de Matem\`{a}tica Multidisciplin\`{a}ria} \\
{\small Universitat Polit\`{e}cnica de Val\`{e}ncia} \\
{\small Cam\'{\i} de Vera s/n. 46022 Val\`{e}ncia, Spain.}\\
 \\
{\small $^3$\,Centro de Tecnolog\'{\i}a Biom\'edica} \\
{\small Universidad Polit\'{e}cnica de Madrid} \\
{\small 28223 Pozuelo de Alarc\'on (Madrid), Spain.}
}

\maketitle

\begin{abstract}
In this paper new results on personalized PageRank are shown. We consider directed graphs that may contain dangling nodes. The main result presented gives an analytical characterization of all the possible values of the personalized PageRank for any node.We use this result to give a theoretical justification of a recent model that uses the personalized PageRank to classify users of Social Networks Sites. We introduce new concepts concerning competitivity and leadership in complex networks. We also present some theoretical techniques to locate leaders and competitors which are valid for any personalization vector and by using only information related to the adjacency matrix of the graph and the distribution of its dangling nodes.
\end{abstract}

\section{Introduction}\label{sec:intro}

Much effort has been done in some aspects related to PageRank and its applications since the introduction of the PageRank algorithm to rank pages on the web \cite{PaBrMoWi}. We are interested in the use of the so-called personalization vector to bias the PageRank to some nodes. We refer the reader to  \cite{LaMebook} and \cite{BoSaVi} for the theoretical basis of the PageRank algorithm.

The idea of biasing the PageRank vector using a personalization vector was, in fact, suggested originally in  \cite{PaBrMoWi}. The first time that someone uses the personalization vector to bias to some topics appears in  \cite{Ha03}. In \cite{HaKaJe} some different ways of biassing the PageRank with personalization vectors are summarized. In \cite{EiVa} the authors propose to use the personalization vector to bias the PageRank to pages that were visited more frequently by previous users. To reduce computational complexity the usual strategies consist in taking low rank approximations \cite{Tong} or to decompose into subgraphs \cite{JeWi}, \cite{VaChGu}. Another approach consists in using Monte Carlo methods to compute only the top-k Personalized PageRank \cite{Avra10}.

As a centrality measure, Personalized PageRank can also be used to classify users in Social Network Sites \cite{Pe}, \cite{Pe11}, \cite{Pe12}. A way of using the Personalized PageRank to rank nodes in an SNS is by using the personalization vector to incorporate features of the users. These features can be popularity
(e.g., the number of friends or followers of the user), activity (e.g., the number of actions made by the user) and recentness (e.g., the up-to-date of the actions of the user in the SNS). In \cite{Pe}, the fundamentals of a model that uses the Personalized PageRank to rank users is SNS were presented. This model left open some theoretical questions. One of these questions is to what extend one can use the personalization vector to modify the PageRank vector. In this paper we address this question and we show, in particular, that the component $i$ of the PageRank vector is bounded, being this bound a sharp one, and valid for any personalization vector.  We also derive some theoretical properties that let us detect nodes that compete with each other to gain PageRank.  We also give
some generalizations to some definitions introduced in \cite{Pe}.

The structure of the paper is as follows. In section~\ref{sec:def} the basic definitions and results used in the rest of the paper are presented. In addition to this, a technical general lemma about row stochastic matrices is proved in order to provide the tools of the results of the following sections. Section~\ref{sec:mainthm} is devoted to prove the main result of the paper that locates all the possible values of the personalized PageRank for each node of a network. Finally,  section~\ref{sec:applications} presents several applications of the localization theorem proved before. The applications include an analytical result that gives necessary and sufficient conditions for the competitivity between nodes  and a characterization of the leadership of nodes in a complex network. These analytical results give some easy algorithms to locate leaders and competitors which are valid for any personalization vector and only use information related to the adjacency matrix of the graph and the distribution of its dangling nodes. In this final section several examples are presented in order to illustrate the results proved.

 \section{Some definitions and a technical lemma}\label{sec:def}

Let ${\mathcal G}= ({\mathcal N}, {\mathcal E})$ be a directed graph where ${\mathcal N} = \{ 1,2, \ldots, n \}$ and $n\in \mathbb{N}$. Note that all the results presented in this paper deal with directed networks, but they can be straightforwardly stated for un-directed networks. The link  $(i,j)$ belongs to the set  ${\mathcal E}$ if and only if there exists a link connecting node $i$ to node $j$. The {\it adjacency matrix} of ${\mathcal G}$ is an $n\times n$-matrix
\[
A=(a_{ij})  \hbox{ where } a_{ij}=\left\{
       \begin{array}{ll}
         1, & \hbox{if $(i,j)$ is a link of ${\mathcal G}$} \\
         0, & \hbox{otherwise.}
       \end{array}\right.
\]
A link $(i,j)$ is said to be an {\em outlink} for node $i$ and an {\em inlink} for node $j$. We denote $k_{out}(i)$ the {\it outdegree} of node $i$, i.e.,  the number of outlinks of a node $i$. Notice that $k_{out}(i)=\sum_k a_{ik}$. The graph ${\mathcal G}= ({\mathcal N}, {\mathcal E})$ may have {\it dangling nodes}, which are nodes $i\in{\mathcal N}$ with zero outdegree.

Let  $P=(p_{ij}) \in\mathbb{R}^{n\times n}$  be the {\it row stochastic matrix} associated to ${\mathcal G}$ defined in the following way:
\begin{itemize}
\item if $i$ is a dangling node, $p_{ij}=0$ for all $j=1,\dots, n$,
\item otherwise, $p_{ij}=\frac{a_{ij}}{k_{out}(i)}=\frac{a_{ij}}{\sum_{k}a_{ik}}$.
\end{itemize}

Vectors of $\mathbb{R}^n$ will be denoted by column matrices. In particular,
\[
\begin{split}
\mathbf{e}_1=&(1,0,\cdots,0)^T\in\mathbb{R}^n,\\
&\cdots \\
\mathbf{e}_n=&(0,\cdots,0,1)^T\in\mathbb{R}^n,\\
\mathbf{e}=&\mathbf{e}_1+\cdots+\mathbf{e}_n=(1,\cdots,1)^T.
\end{split}
\]
The $i^{\rm th}$-component of a vector $\mathbf{v}=(v_1,\cdots, v_n)^T\in\mathbb{R}^n$
is given by the product $\mathbf{v}^T\mathbf{e}_i=v_i$, and the sum of the components of the vector $\mathbf{v}\in \mathbb{R}^n$ is 1 if  $\mathbf{v}^T\mathbf{e}=1$. Moreover, we will say that $\mathbf{v}>0$ if all the components $v_i$ of $\mathbf{v}$ are greater than zero, i.e., $\mathbf{v}^T\mathbf{e}_i>0$, $i=1,\dots, n$.

We will use the {\it personalized PageRank vector}. The ingredients to build such PageRank vector \cite{PaBrMoWi} are:
\begin{itemize}
\item A {\it damping factor} $\alpha\in (0,1)$.
\item A {\it distribution of dangling nodes} $\mathbf{u}\in \mathbb{R}^n
$ such that $\mathbf{u}>0$  and $\mathbf{u}^T\mathbf{e}=1$. The dangling nodes will be characterized by a vector $\mathbf{ d}\in \mathbb{R}^n$  defined as $\mathbf{d}=(d_1,\dots,d_n)^T$, where
\[
d_i=\left\{
       \begin{array}{ll}
         1, & \hbox{if $i$ is a dangling node of ${\mathcal G}$} \\
         0, & \hbox{otherwise.}
       \end{array}
     \right.
\]
\item A {\it personalization  vector} $\mathbf{v}\in\mathbb{R}^n$ such that $\mathbf{v}>0$  and $\mathbf{v}^T\mathbf{e}=1$.
\end{itemize}
The {\it Google matrix} $G=G(\alpha,\mathbf{u},\mathbf{v})$ with dangling nodes and personalized vector $\mathbf{v}$ is defined as
\begin{equation}\label{e:PRMatrix}
G=\alpha(P+\mathbf{d}\mathbf{u}^T)+(1-\alpha)\mathbf{e}\mathbf{v}^T\in \mathbb{R}^{n\times n}
\end{equation}
(this matrix $G$ is row-stochastic, i.e., $G\mathbf{e}=\mathbf{e}$). The {\it PageRank vector} $\mathbf{\pi}=\mathbf{ \pi}(\alpha,\mathbf{u},\bf{v})$ is the unique eigenvector of $G^T$ associated to eigenvalue 1 such that $\mathbf{\pi}^T\mathbf{ e}=1$, i.e., $\mathbf{\pi}>0$, $\mathbf{\pi}^T\mathbf{e}=1$ and $\mathbf{\pi}^TG=\mathbf{\pi}^T$ (see \cite{PaBrMoWi}).

From now on we will consider a fixed damping factor $\alpha$ (usually $\alpha=0.85$) and a fixed distribution of dangling nodes $\textbf{u}$, so the PageRank matrix $G$ and the PageRank vector will only depend on the personalization vector $\mathbf{v}$ we are considering.

Since $\mathbf{\pi}^TG=\mathbf{\pi}^T$, from the definition of $G$ in (\ref{e:PRMatrix}) and the fact that $\mathbf{\pi}^T\mathbf{e}=1$ we get that
\[
\begin{split}
\mathbf{\pi}^T&=\mathbf{\pi}^TG=\mathbf{\pi}^T\left(\alpha(P+\mathbf{d}\mathbf{u}^T)+(1-\alpha)\mathbf{e}\mathbf{v}^T\right)\\
&=\alpha\mathbf{\pi}^T(P+\mathbf{d}\mathbf{u}^T)+(1-\alpha)\mathbf{\pi}^T\mathbf{e}\mathbf{v}^T\\
&=\alpha\mathbf{\pi}^T(P+\mathbf{d}\mathbf{u}^T)+(1-\alpha)\mathbf{v}^T
\end{split}
\]
so $\mathbf{\pi}^T(I_n-\alpha P-\alpha \mathbf{d}\mathbf{u}^T)=(1-\alpha)\mathbf{v}^T$, where $I_n\in \mathbb{R}^{n\times n}$ is the identity matrix. Therefore it was shown in \cite{BoSaVi} that
\begin{equation}\label{e:Boldhi}
\mathbf{\pi}^T=(1-\alpha)\mathbf{v}^T(I_n-\alpha(P+\mathbf{d}\mathbf{u}^T))^{-1}.
\end{equation}
We will denote by $X$ the $n\times n$-matrix  appearing in formula (\ref{e:Boldhi}) above
\[
X=(1-\alpha)(I_n-\alpha(P+\mathbf{d}\mathbf{ u}^T))^{-1}, \hbox{ so } \mathbf{\pi}^T=\mathbf{v}^TX.
\]
Notice that formula (\ref{e:Boldhi}) indicates that the PageRank of each node can be expressed as a function of the personalization vector $\mathbf{v}$ since $\mathbf{\pi}=\mathbf{\pi}(\mathbf{v})=X^T\mathbf{v}$ (see \cite{BoSaVi}). Notice that this equality  makes sense for all $\mathbf{v}\in \mathbb{R}^n$ and gives the PageRank when $\mathbf{v}>0$ and $\mathbf{v}^T\mathbf{e}=1$.

It is easy to check that the matrix $P_{\mathbf{u}}=P+\mathbf{d}\mathbf{u}^T$ appearing in (\ref{e:Boldhi}) is a row-stochastic matrix since
\[
P_{\mathbf{u}}\mathbf{e}= \left(P+\mathbf{d}\mathbf{u}^T\right)\mathbf{e}
=P\mathbf{e}+\mathbf{d}\mathbf{u}^T\mathbf{e}
=P\mathbf{e}+\mathbf{d}=\mathbf{e}.
\]
In the next section we will use the following lemma dealing with row-stochastic matrices as $P_{\mathbf{u}}$:

\begin{lemma}\label{diagonaldominant}
Let $Q$ be a row-stochastic matrix and $\alpha\in(0,1)$. Then the matrix $Y=I_n-\alpha Q$ is strictly row-diagonally-dominant,  $X=(1-\alpha)Y^{-1}$ is  strictly diagonally-dominant of its column entries and the maximum of each column $i$ of $X$ is achieved in $x_{ii}$.
\end{lemma}

\begin{proof}
Clearly $Q\mathbf{e}=I_n\mathbf{e}-\alpha Q\mathbf{e}=(1-\alpha)\mathbf{e}$, i.e., the sum of the entries of each row of $Q$ is $1-\alpha$. Therefore, since $\alpha\in (0,1)$ and $0\le q_{ik}\le 1$ for all $i,k=1,\dots, n$, we get that
\[
|y_{ii}|=|1-\alpha q_{ii}|=1-\alpha q_{ii}=1-\alpha+\alpha\sum_{k\ne i} q_{ik}>\alpha\sum_{k\ne i} q_{ik}=\sum_{k\ne i}|y_{ik}|,
\]
i.e., $Y$ is strictly row-diagonally-dominant. Now, by Theorem 2.5.12 in \cite{HoJo}, $Y^{-1}$ and $X=(1-\alpha)Y^{-1}$ are  strictly diagonally-dominant of their column entries and therefore for every $k\ne i$
\[
|x_{ii}|>|x_{ki}|.
\]
Moreover, since $Y$ is a (nonsingular) M-matrix (see, for example, \cite{BP}) we have that $Y^{-1} \geq 0$. Hence the absolute values in the formula above can be deleted and we get
\[
\max_k x_{ki}=x_{ii}.
\]
\end{proof}

\section{Main result: Location of Personalized PageRank}\label{sec:mainthm}

The main contribution of this paper is the solution to the following problem:
\medskip

\noindent\textbf{Problem.} Given a graph ${\mathcal G}$  with dangling nodes indicated by some vector $\mathbf{d}$, a fixed damping factor $\alpha\in (0,1)$ and fixed dangling nodes distribution $\mathbf{u}$, is there an easy way to locate all the possible values of the PageRank for each node $i$?

\begin{definition}{} Given a graph ${\mathcal G}$  with dangling nodes indicated by some vector $\mathbf{d}$, a fixed damping factor $\alpha\in (0,1)$ and fixed dangling nodes distribution $\mathbf{u}$, for each node $i\in{\mathcal N}$ we define ${\mathcal {PR}}(i)$ as the set of all possible values of Personalized PageRank of node $i$, i.e.,
$$
{\mathcal {PR}}(i)=\{\pi^T(\mathbf{v})\mathbf{e}_i \hbox{ for all $\mathbf{v}\in\mathbb{R}^n$},\ \mathbf{v}>0,\ \mathbf{v}^T\mathbf{e}=1\}\subset (0,1).
$$
\end{definition}

The following theorem shows that ${\mathcal {PR}}(i)$ coincides with an open interval whose extreme values are given by the the maximum and minimum entries of the $i^{\rm th}$-column of $X$.

\begin{theorem}\label{mainth} Given a graph ${\mathcal G}$  with dangling nodes indicated by some vector $\mathbf{d}$, a fixed damping factor $\alpha\in (0,1)$ and fixed dangling nodes distribution $\mathbf{u}$, for each node $i\in{\mathcal N}$
\[
{\mathcal {PR}}(i)=(\min_j x_{ji},\ x_{ii}),
\]
where $X=(x_{ij})=(1-\alpha)(I_n-\alpha(P-\mathbf{d}\mathbf{u}^T))^{-1}$ is the matrix appearing in formula (\ref{e:Boldhi}).
\end{theorem}

\begin{proof}
We will separate the proof of the theorem in two steps:
\begin{enumerate}
\item[{Step 1.}] $\min_j x_{ji}<{\mathcal {PR}}(i)< x_{ii}$ for every personalization vector $\mathbf{v}$;
\item[{Step 2.}] every $x$ with $\min_j x_{ji}< x< x_{ii}$ can be achieved as the PageRank of node $i$ for a certain personalization vector $\mathbf{v}$.
\end{enumerate}

\noindent\underbar{Proof of Step 1.} Without loss of generality we can suppose that $i=1$. Let $\mathbf{v}\in \mathbb{R}^n$ such that $\mathbf{v}^T\mathbf{e}=1$. Then the first component of $\pi(\mathbf{v})$ is
\[
\pi^T(\mathbf{v})\mathbf{e}_1=\mathbf{v}^TX\mathbf{e}_1=\mathbf{v}^T\left(
                                                                           \begin{array}{c}
                                                                             x_{11} \\
                                                                             \vdots \\
                                                                             x_{n1} \\
                                                                           \end{array}
                                                                         \right)=\sum_j v_jx_{j1}.
\]
In particular, if $\mathbf{v}$ is a personalization vector ($\mathbf{v}>0$ and $\mathbf{v}^T\mathbf{e}=1$),  $\pi(\mathbf{v})$ is the PageRank corresponding to this personalization vector and the formula above gives the first component of the PageRank. Since in this case all the components of  $\mathbf{v}$ are positive and $\sum_j v_j=1$, $\sum_j v_jx_{j1}$ is a {\it strict} convex combination of the entries of the first column of $X$ and $\min_j x_{j1}<\sum_j v_jx_{j1}<\max_j x_{j1}$. Moreover, by Lemma \ref{diagonaldominant}, $\max_j x_{j1}=x_{11}$ and Step 1 is shown.

\medskip
\noindent\underbar{Proof of Step 2.}  Without loss of generality  suppose again that $i=1$. By the calculations done in Step 1, the first component of  $\pi(\mathbf{v})$ of every $\mathbf{v}\in\mathbb{R}^n$, $\mathbf{v}^T\mathbf{e}=1$, equals $\sum_j v_jx_{j1}$. In particular, the first component of  $\pi(\mathbf{e}_1)$ is $x_{11}$, the first component of  $\pi(\mathbf{e}_2)$ is $x_{21}$, etc., and the extreme values of the open interval
$$
{\mathcal {PR}}(1)=(\min_j x_{j1},\ x_{11})
$$
would be achieved if we admitted  $\mathbf{e}_1$ and $\mathbf{e}_{j_1}$, where we denote by $j_1$ an index where the minimum of the first column of $X$ is reached.

Now we define
\[
\mathbf{v}_{1\varepsilon}=\left(\begin{array}{c}
                                        1-\varepsilon \\
                                        \frac{\varepsilon}{n-1}                                        \\
                                        \frac{\varepsilon}{n-1}                                        \\
                                        \vdots \\
                                        \frac{\varepsilon}{n-1}
                                        \end{array}
                                        \right)\qquad
\mathbf{v}_{j_1\varepsilon}=\left(\begin{array}{c}
                                        \frac{\varepsilon}{n-1} \\
                                        \vdots \\
                                        1-\varepsilon\\
                                        \vdots \\
                                        \frac{\varepsilon}{n-1}
                                        \end{array}
                                        \right)\leftarrow j_1\text{ coordinate}
\]
for every $\varepsilon\in(0,1)$. Then it is easy to check that both $\mathbf{v}_{1\varepsilon}>0$ and $\mathbf{v}_{j_1\varepsilon}>0$, $\mathbf{ v}_{1\varepsilon}^T\mathbf{e}=1=\mathbf{v}_{j_1\varepsilon}^T\mathbf{e}$, and
\[
\lim_{\varepsilon\to 0^+}\pi^T(\mathbf{v}_{1\varepsilon})\mathbf{e}_1=x_{11},\qquad
\lim_{\varepsilon\to 0^+}\pi^T(\mathbf{v}_{j_1\varepsilon})\mathbf{e}_1=x_{j_11}.
\]
Finally, for every $\lambda\in(0,1)$ we define
\[
\mathbf{v}_{\lambda\varepsilon}=\lambda\mathbf{ v}_{1\varepsilon}+(1-\lambda)\mathbf{v}_{j_1\varepsilon}>0
\]
which satisfies that
\[
\lim_{\varepsilon\to 0^+} \pi^T(\mathbf{ v}_{\lambda\varepsilon})\mathbf{e}_1=\lambda x_{11}+(1-\lambda) x_{j_11},\hbox{ so}
\]
\[
\lim_{\lambda\to 1}\lim_{\varepsilon\to 0^+} \pi^T(\mathbf{v}_{\lambda\varepsilon})\mathbf{e}_1=x_{11}
\]
\[
\lim_{\lambda\to 0}\lim_{\varepsilon\to 0^+} \pi^T(\mathbf{v}_{\lambda\varepsilon})\mathbf{e}_1=x_{j_11}
\]
and hence for every $x$ with $x_{j_11}< x< x_{11}$  there exists some $\varepsilon_0,\lambda_0\in (0,1)$ such that
\[
\pi^T(\mathbf{ v}_{\lambda_0\varepsilon_0})\mathbf{e}_1=x.
\]
\end{proof}

\begin{remark}
Note that if $\mathbf{v}>0$ such that $\mathbf{v}^T\mathbf{e}=1$, then for every $i\in\mathcal{N}$
$
\pi^T(\mathbf{v})\mathbf{e}_i
$,
is the PageRank of node $i$ when using the personalization vector $\mathbf{v}$. A convenient notation for this value is $PR(i,\mathbf{v})=\pi^T(\mathbf{v})\mathbf{e}_i$. We note here that the personalization vectors considered in \cite{Pe} were of the form ${\bf v}_{j\varepsilon}$ and there the definitions  only  deal with  $PR(i,\mathbf{v}_{j\varepsilon})$. The competitivity interval in the sense of \cite{Pe} is defined as
\[
S_{C}(i,\varepsilon) = [ \min_{j \in \cal{N}} PR(i,\mathbf{v}_{j\varepsilon}), \max_{j \in \cal{N}} PR(i,\mathbf{v}_{j\varepsilon})],
\]
for every $\varepsilon\in(0,1)$ and each $i\in\mathcal{N}$. Then from Theorem \ref{mainth} it is clear that for a given $\varepsilon \in (0,1)$ and each $i \in  \cal{N}$ we have
\[
S_{C}(i,\varepsilon) \subset \mathcal{PR}(i) = \bigcup_{\varepsilon > 0} S_{C}(i,\varepsilon).
\]
\end{remark}

\section{Some applications}\label{sec:applications}

In addition to the intrinsic interest of the previous results, the techniques developed in the last section can be useful in order to analyze the competitivity of nodes in a network according to their Personalized PageRank and other problems such as the localization of leaders in a complex network. It is well known that Personalized PageRank is a very remarkable tool that helps ranking the nodes of a network according to their centrality (see, for example, \cite{PaBrMoWi, BoSaVi, Pe}). This main fact makes that in many real-life networks (such as WWW networks or social networks) it is crucial for a node $i$ to spot other nodes that can be overcome by $i$ in a ranking based on Personalized PageRank, since these nodes are the nodes that actually {\em compete} with $i$ in the ranking based on Personalized PageRank. This problem has already been considered in the literature (see, for example, \cite{Pe}). The techniques developed in the previous section can give a computationally efficient solution to the characterization of the competing nodes of a fixed vertex $i$. Let us start stating the basic definition of {\em competitivity} between two nodes in a complex network.

\begin{definition}
Given two nodes $i,j$ ($i\ne j$) of a graph ${\mathcal G}=({\mathcal N}, {\mathcal E})$, we say that $i$ and $j$ are {\it effective competitors} if there exist two personalization vectors ${\bf v},{\bf w}$ (${\bf v},{\bf w}>0$ and ${\bf v}^T{\bf e}=1={\bf w}^T{\bf e}$) such that the $i^{\rm th}$-component of the personalized PageRank with respect to ${\bf v}$ is greater than the $j^{\rm th}$-component of the personalized PageRank with respect to ${\bf v}$, but the $i^{\rm th}$-component of the personalized PageRank with respect to ${\bf w}$ is smaller than the $j^{\rm th}$-component of the personalized PageRank with respect to ${\bf w}$, i.e.,
\begin{align*}
\pi^T({\bf v}){\bf e}_i&>\pi^T({\bf v}){\bf e}_j\\ \pi^T({\bf w}){\bf e}_i&<\pi^T({\bf w}){\bf e}_j.
\end{align*}
This definition means that nodes $i$ and $j$ appear with different rank in the personalized PageRank vector if we consider some different personalization vectors ${\bf v}$ and ${\bf w}$.
\end{definition}

\begin{remark}
Note that this definition is more restrictive than the definition of {\em competitivity group} given in \cite{Pe}. Furthermore note that the fact of being in the same {\em competitivity group} is a necessary but not a sufficient condition to be effective competitors. Later we show some examples of this fact.
\end{remark}

We will see in this section that the results and techniques coming from the last section give a positive answer to the following question:

\noindent{\bf Question.} Is there an easy method of knowing whether two given nodes are effective competitors or not?

This question and this kind of problems have been posed in the literature in social networks, and actually in \cite{Pe} a necessary condition  for a couple of nodes $i,j\in \mathcal{N}$ to compete  is given in terms of the so-called {\em competitivity intervals}. We will see in example \ref{ex:competitivity} that the result used in \cite{Pe} only gives necessary conditions for competitivity between nodes, while the following result gives a complete characterization of the competitors of a given node.

\begin{theorem}\label{th:competitivity}
Given a graph ${\mathcal G}=({\mathcal N}, {\mathcal E})$  with dangling nodes indicated by some vector ${\bf d}$, a fixed damping factor $\alpha\in (0,1)$ and fixed dangling nodes distribution ${\bf u}$, two nodes $i,j\in {\mathcal N}$ are effective competitors if and only if there exist $k,\ell\in \{1,\dots, n\}$ such that
\[
x_{ki}>x_{kj}\quad\hbox { and }\quad x_{\ell i}<x_{\ell j},
\]
where $X=(x_{pq})=(1-\alpha)(I_n-\alpha(P+{\bf d}{\bf u}^T))^{-1}$ is the $n\times n$-matrix given in formula (\ref{e:Boldhi}).
\end{theorem}

\begin{proof}
If we consider ${\bf v}_{k\varepsilon}$ and ${\bf v}_{ \ell\varepsilon}$  as defined in the proof of theorem \ref{mainth},
\begin{align*}
&\lim_{\varepsilon\to 0}\pi^T({\bf v}_{k\varepsilon}){\bf e}_i=x_{ki},\\
&\lim_{\varepsilon\to 0}\pi^T({\bf v}_{k\varepsilon}){\bf e}_j=x_{kj}, \\
&\lim_{\varepsilon\to 0}\pi^T({\bf v}_{\ell\varepsilon}){\bf e}_i=x_{\ell i},\\
&\lim_{\varepsilon\to 0}\pi^T({\bf v}_{\ell\varepsilon}){\bf e}_j=x_{\ell j},
\end{align*}
so from $x_{ki}>x_{kj}$  and $x_{\ell i}<x_{\ell j}$ there exists $\varepsilon>0$ such that the choice of  ${\bf v}_{k\varepsilon}$ or ${\bf v}_{\ell\varepsilon}$ as personalization vectors exchanges the order of nodes $i$ and $j$ in the PageRank vector with  respect to such personalization vectors.

Conversely, suppose that $i$ and $j$ are effective competitors but for all $k\in\{1,\dots, n\}$ $x_{ki}\ge x_{kj}$ (similarly, $x_{ki}\le x_{kj}$). By hypothesis, there exist some  personalization vectors ${\bf v},{\bf w}$ such that $\pi^T({\bf v}){\bf e}_i>\pi^T({\bf v}){\bf e}_j$ and $\pi^T({\bf w}){\bf e}_i<\pi^T({\bf w}){\bf e}_j.
$ In particular, if ${\bf w}^T=(w_1,\dots, w_n)$,
\[
\pi^T({\bf w}){\bf e}_i<\pi^T({\bf w}){\bf e}_j=\sum_\ell w_\ell x_{\ell j}\le\sum_\ell w_\ell x_{\ell i}=\pi^T({\bf w}){\bf e}_i
\]
leading to a contradiction.
\end{proof}

\begin{remark}
This theorem gives an easy way to search for effective competitors: it is enough to compare the $i^{\rm th}$-column and the $j^{\rm th}$-column of matrix $X$; if each entry of the $i^{\rm th}$-column is always greater or equal than the corresponding entry of the $j^{\rm th}$-column (or if it is always smaller or equal), then nodes $i$ and $j$ are not effective competitors. Otherwise, some change in the sign of the difference between columns $i$ and $j$ provide the existence of effective competitors. Moreover, if the  changes of sign occur in rows $k$ and $\ell$, we can assure that there exists $\varepsilon>0$ such that ${\bf v}_{k\varepsilon}$ and ${\bf v}_{\ell\varepsilon}$ are personalization vectors that make nodes $i$ and $j$ compete.
\end{remark}

Let us present an example of the use of the previous result and how the intersection condition presented in \cite{Pe} gives less information than the corresponding one obtained from theorem \ref{th:competitivity}.

\begin{example}\label{ex:competitivity}
Let us consider the network $\mathcal{G}_1=(V_1,E_1)$ given in figure \ref{figure:01}.
\begin{figure}[h!]
\centering
\includegraphics[scale=0.3]{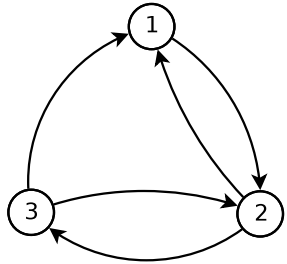}
\caption{A directed network $\mathcal{G}_1=(V_1,E_1)$ with 3 nodes}
\label{figure:01}
\end{figure}

The adjacency matrix of $\mathcal{G}_1$ is
\[
A_1=\left(
\begin{array}{ccc}
0 & 1 & 0 \\
1 & 0 & 1 \\
1 & 1 & 0
\end{array}
\right)
\]
Then, since $\mathcal{G}_1$ has no dangling nodes, if we fix $\alpha=0.85$, then we can compute the matrix $(x_{ij})=X_1=(1-\alpha)\left(I_3-\alpha P\right)^{-1}$ obtaining
\[
X_1=\left(
\begin{array}{ccc}
0.4035 & 0.4186 & 0.1779 \\
0.2982 & 0.4925 & 0.2093 \\
0.2982 & 0.3872 & 0.3146
\end{array}
\right).
\]
Hence, by using theorem \ref{mainth} we get that
\[
\begin{split}
{\mathcal {PR}}(1)&=(0.2982,0.4035),\\
{\mathcal {PR}}(2)&=(0.3872,0.4925),\\
{\mathcal {PR}}(3)&=(0.1779,0.3146).
\end{split}
\]
If we use the necessary conditions obtained in \cite{Pe}, we get that node $1$ could compete with nodes $2$ and $3$, but nodes $2$ and $3$ cannot compete between them since ${\mathcal {PR}}(1)\cap {\mathcal {PR}}(2)\ne \emptyset \ne {\mathcal {PR}}(1)\cap {\mathcal {PR}}(3)$ and ${\mathcal {PR}}(2)\cap{\mathcal {PR}}(3)=\emptyset$.

In addition to this, if we use the criterion given in theorem \ref{th:competitivity}, we note that while $1$ and $3$ are actually effective competitors, nodes $1$ and $2$ do not compete. Indeed, by comparing on the one hand the first with the third column of $X_1$  we get that $x_{11}>x_{13}$ and $x_{21}<x_{33}$ while, on the other hand, by comparing the first and the second columns of $X_1$ we get that $x_{i1}<x_{i2}$ for all $i=1,2,3$.
\end{example}

Another type of problems that can be solved by using the techniques introduced in the previous section deal with leadership of nodes. The leadership in complex networks has been studied in the Complex Networks Analysis from very different points of view, including (among others) the use of efficiency and robustness perspective in networks related with cryptography (see \cite{CrFlGVPe}) and Personalized PageRank in social networks (see \cite{Pe}). Roughly speaking a node $i$ is a leader (for the personalized PageRank-based ranking) if its personalized PageRank is maximal among all the nodes of the network for some personalization vector. This concept was studied in \cite{Pe} but only considering personalization vectors of the form  ${\bf v}_{j\varepsilon}$. As an extension of this concept we introduce the following definition.

\begin{definition}
 Given a node $i$ of a graph ${\mathcal G}=({\mathcal N}, {\mathcal E})$, we say that $i$ is a {\it leader} of ${\mathcal G}$ if there exists a personalization vector  $\mathbf{v}\in\R^n$ ($\mathbf{v}>0$ and $\mathbf{v}^T\mathbf{e}=1$) such that for every node $j\in \mathcal N$ ($j\ne i$)
 \[
 \pi^T({\bf v}){\bf e}_i>\pi^T({\bf v}){\bf e}_j.
 \]
 The set of all leader nodes of a graph $\mathcal G$ is called the {\it leadership group} of the network.
\end{definition}

Once we have considered the definition of the leadership group of a graph $\mathcal G$ it is natural to ask the following question:

\noindent{\bf Question.} Is there an easily-computable way to determine the leadership group of a graph $\mathcal G$?

Once more this question was considered in \cite{Pe} and some results in terms of competitivity intervals were presented, but they only gave sufficient conditions for a node $i$ to be a leader of the network. By using our methods we can go further and prove the following result:

\begin{theorem}\label{th:leadership}
Given a graph ${\mathcal G}=({\mathcal N}, {\mathcal E})$  with dangling nodes indicated by some vector ${\bf d}$, a fixed damping factor $\alpha\in (0,1)$ and fixed dangling nodes distribution ${\bf u}$, the leadership group of $\mathcal G$ is the set of nodes $i\in \mathcal N$ verifying that  there is a value $ j\in \mathcal N$ such that for every $k\in \mathcal N$ ($k\ne i$)
\[
x_{ji}>x_{jk},
\]
where $X=(x_{pq})=(1-\alpha)(I_n-\alpha(P+{\bf d}{\bf u}^T))^{-1}$ is the $n\times n$-matrix given in formula (\ref{e:Boldhi}).
\end{theorem}

\begin{proof}
Let us denote
\[
\begin{split}
A&=\left\{i\in \mathcal{N};\enspace i\text{ is a leader of }\mathcal{G}\right\},\\
B&=\left\{i\in \mathcal{N};\enspace \exists j \in {\cal N}: x_{ji}>x_{jk} \text{ for all }k\ne i  \right\}.
\end{split}
\]
On the one hand, if we take $i\in B$, since there is a value $j \in {\cal N}$ such that for every $k\ne i$ we get that $x_{ji}>x_{jk}$, by using the same techniques (and notation) as in the proof of theorem \ref{mainth}, a simple  continuity argument makes that there is an $\varepsilon\in (0,1)$ such that for every $k\ne i$
\[
\pi^T(\textbf{v}_{j\varepsilon})\textbf{e}_i>\pi^T(\textbf{v}_{j\varepsilon})\textbf{e}_k,
\]
which makes that $i\in A$ and therefore $A\subseteq B$.

On the other hand, if $i\in A$, there is a personalization vector $v\in\R$ such that for every $k\ne i$ we know that $\pi^T({\bf v}){\bf e}_i>\pi^T({\bf v}){\bf e}_k$. As it was proved in step 1 of the proof of theorem \ref{mainth}, if we denote $v=(v_1,\cdots,v_n)$ we get that  for every $k\ne i$
\begin{equation}\label{e:contra}
\sum_{j}v_jx_{ji}=\pi^T({\bf v}){\bf e}_i> \pi^T({\bf v}){\bf e}_k=\sum_{j}v_jx_{jk}.
\end{equation}
We are going to show that if the last expression holds, then  there is a value $j \in {\cal N}$ such that $x_{ji}>x_{jk}$ for all $k\ne i$, otherwise for every  $j \in {\cal N}$ it should be that $x_{ji}\le x_{jk}$ and hence for every  $k\ne i$
\[
\pi^T({\bf v}){\bf e}_i= \sum_{j}v_jx_{ji}\le \sum_{j}v_jx_{jk}= \pi^T({\bf v}){\bf e}_k,
\]
which contradicts equation (\ref{e:contra}). Therefore there is a value $1\le j\le |\mathcal{N}|$ such that $x_{ji}>x_{jk}$ for all $k\ne i$, which makes that $i\in B$ and hence $B\subseteq A$ and we conclude the proof.
\end{proof}

\begin{remark}
Note that the last result gives an effective algorithm to locate the leadership group of a network ${\mathcal G}=({\mathcal N}, {\mathcal E})$. It is enough to spot the maximum of each row of the $n\times n$-matrix $X$  given in formula (\ref{e:Boldhi}) and if it takes place at the element $x_{ij}$, then the node $j$ is a leader of the network. Actually, the last result ensures that the only possible leader nodes of the networks are those who fulfil this property.

Therefore the matrix $X=(x_{pq})=(1-\alpha)(I_n-\alpha(P+{\bf d}{\bf u}^T))^{-1}$ encapsulates a lot of useful information about the personalized PageRank of the network: The extremal values on each column $i$ correspond to the extremal values of the set ${\mathcal {PR}}(i)$, the comparison between two columns $j\ne k$ gives the information about the effective competitivity between $j$ and $k$ and finally, the maximum on each row gives a leader node of the graph.

\end{remark}

Let us finish this section with a couple of examples that illustrate the use of the last result.

\begin{example}
Let us take the network $\mathcal{G}_2=(V_2,E_2)$ given in figure \ref{figure:02}.
\begin{figure}[h!]
\centering
\includegraphics[scale=0.3]{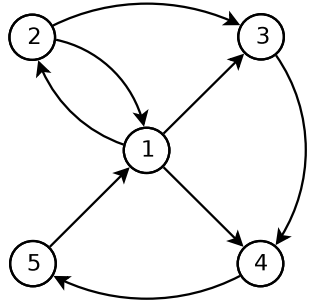}
\caption{A directed network $\mathcal{G}_2=(V_2,E_2)$ with 5 nodes}
\label{figure:02}
\end{figure}

The adjacency matrix of $\mathcal{G}_2$ is
\[
A_2=\left(
\begin{array}{ccccc}
     0  &   1  &   1  &   1  &   0 \\
     1  &   0  &   1  &   0  &   0 \\
     0  &   0  &   0  &   1  &   0 \\
     0  &   0  &   0  &   0  &   1 \\
     1  &   0  &   0  &   0  &   0
\end{array}
\right)
\]
Then, since $\mathcal{G}_2$ has no dangling nodes, if we fix $\alpha=0.85$, then we can compute the matrix $X_2=(1-\alpha)\left(I_5-\alpha P\right)^{-1}$ obtaining
\[
X_2=\left(
\begin{array}{ccccc}
    0.3514 &   0.0995 &   0.1419 &   0.2201 &   0.1871 \\
    0.2410 &   0.2183 &   0.1611 &   0.2052 &   0.1744 \\
    0.2158 &   0.0611 &   0.2371 &   0.2627 &   0.2233 \\
    0.2539 &   0.0719 &   0.1025 &   0.3090 &   0.2627 \\
    0.2986 &   0.0846 &   0.1206 &   0.1871 &   0.3090
\end{array}
\right).
\]
Now, theorem \ref{mainth} determines the set of all possible personalized PageRank values of all the nodes and we get that
\[
\begin{split}
{\mathcal {PR}}(1)&=(0.2158,0.3514),\\
{\mathcal {PR}}(2)&=(0.0611,0.2183),\\
{\mathcal {PR}}(3)&=(0.1025,0.2371),\\
{\mathcal {PR}}(4)&=(0.1871,0.3090),\\
{\mathcal {PR}}(5)&=(0.1744,0.3090).
\end{split}
\]

Note that the maximum of the first and second column is reached at $x_{11}$ and $x_{21}$ respectively. The maximum of the third and forth column is reached at $x_{34}$ and $x_{44}$ respectively, and the maximum of the last row is reached at $x_{55}$. Therefore, by using the method given by theorem \ref{th:leadership} we get that the leadership group is $\{1,4,5\}$.
\end{example}

\begin{example}
Finally, let us now take the network $\mathcal{G}_3=(V_3,E_3)$ that was introduced in \cite{Yu} (see figure~\ref{figure:03}).
\begin{figure}[h!]
\centering
\includegraphics[scale=0.4]{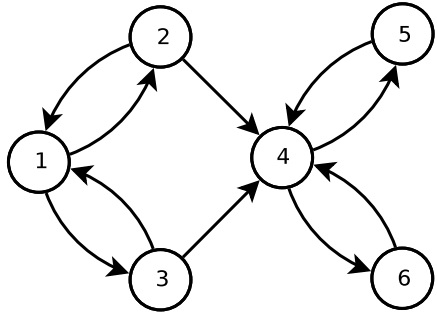}
\caption{A directed network $\mathcal{G}_3=(V_3,E_3)$ with 6 nodes introduced in \cite{Yu}}
\label{figure:03}
\end{figure}

The adjacency matrix of $\mathcal{G}_3$ is
\[
A_3=\left(
\begin{array}{cccccc}
     0  &   1  &   1  &   0  &   0  &   0 \\
     1  &   0  &   0  &   1  &   0  &   0 \\
     1  &   0  &   0  &   1  &   0  &   0 \\
     0  &   0  &   0  &   0  &   1  &   1 \\
     0  &   0  &   0  &   1  &   0  &   0 \\
     0  &   0  &   0  &   1  &   0  &   0
\end{array}
\right)
\]
Then, since $\mathcal{G}_3$ has no dangling nodes, if we fix $\alpha=0.85$, then we can compute the matrix $X_3=(1-\alpha)\left(I_6-\alpha P\right)^{-1}$ obtaining
\[
X_3=\left(
\begin{array}{cccccc}
    0.2348 &   0.0998  &  0.0998 &   0.3057 &   0.1299  &  0.1299 \\
    0.0998 &   0.1924  &  0.0424 &   0.3597 &   0.1529  &  0.1529 \\
    0.0998 &   0.0424  &  0.1924 &   0.3597 &   0.1529  &  0.1529 \\
    0      &   0       &  0      &   0.5405 &   0.2297  &  0.2297 \\
    0      &   0       &  0      &   0.4595 &   0.3453  &  0.1953 \\
    0      &   0       &  0      &   0.4595 &   0.1953  &  0.3453 \\
\end{array}
\right).
\]
In order to determine the set of all possible personalized PageRank values for all nodes, we  use once more  theorem \ref{mainth} and we obtain that
\[
\begin{split}
{\mathcal {PR}}(1)&=(0,0.2348),\qquad{\mathcal {PR}}(4)=(0.3057,0.5405),\\
{\mathcal {PR}}(2)&=(0,0.1924),\qquad{\mathcal {PR}}(5)=(0.1299,0.3453),\\
{\mathcal {PR}}(3)&=(0,0.1924),\qquad{\mathcal {PR}}(6)=(0.1299,0.3453).\\
\end{split}
\]
{}From theorem \ref{th:leadership} we have that the the leadership group is $\{4\}$. We recall here that matrix $G$ given by (\ref{e:PRMatrix}) is always an irreducible matrix since ${\bf v} > 0$. Nevertheless, $X_3$ may be  a reducible matrix as it happens in this example.
\end{example}

%
%
%

\section*{Acknowledgements}

This paper was partially supported by Spanish MICINN Funds and FEDER Funds MTM2009-13848, MTM2010-16153 and MTM2010-18674, Rey Juan Carlos University Funds I3-2010/00075/001 and Junta de Andalucia Funds FQM-264.



\end{document}